\documentclass[11pt]{article}
\usepackage{graphicx} 

\usepackage{fullpage}
\usepackage[utf8]{inputenc}

\usepackage{graphicx}
\usepackage{amsmath,amsthm,amssymb}
\usepackage{algorithm}
\usepackage{algpseudocode}
\usepackage{multirow}
\usepackage{makecell}
\usepackage{mathrsfs}

\usepackage{thm-restate,color,xspace}
\usepackage{comment}
\usepackage{thmtools}
\usepackage{xcolor}
\usepackage{nameref}
\usepackage{array}

\definecolor{ForestGreen}{rgb}{0.1333,0.5451,0.1333}
\definecolor{DarkRed}{rgb}{0.65,0,0}
\definecolor{Red}{rgb}{1,0,0}
\usepackage[linktocpage=true,
pagebackref=true,colorlinks,
linkcolor=DarkRed,citecolor=ForestGreen,
bookmarks,bookmarksopen,bookmarksnumbered]
{hyperref}
\usepackage{cleveref}

\declaretheorem[numberwithin=section]{theorem}
\declaretheorem[numberlike=theorem]{lemma}

\declaretheorem[numberlike=theorem]{fact}

\declaretheorem[numberlike=theorem]{claim}

\global\long\def\Otilde{\widetilde{O}}
\newcommand{\poly}{\operatorname{poly}}
\newcommand{\vecone}{\mathbf{1}}
\newcommand{\F}{\mathbb{F}}
\newcommand{\Zcal}{\mathcal{Z}}

\newcommand{\Bscr}{\mathscr{B}}
\newcommand{\Acal}{\mathcal{A}}
\newcommand{\Bcal}{\mathcal{B}}

\newcommand{\ppt}{\operatorname{pt}}

\title{Deterministic Edge-Fault-Tolerant Connectivity Labeling Schemes\\with Nearly Optimal Label Size}
\author{Yaowei Long\\ University of Michigan
\and 
Seth Pettie\thanks{Supported by NSF Grants CCF-2221980 and CCF-2446604.}
\\ University of Michigan
\and 
Thatchaphol Saranurak\thanks{Supported by NSF Grant CCF-2238138 and Sloan Fellowship}
\\ University of Michigan}

\date{}
\begin{document}

\maketitle

\begin{abstract}
For an undirected graph $G = (V,E)$ and a fault bound $f$, an \emph{edge-fault-tolerant connectivity labeling scheme} assigns short labels to vertices and edges, so that for any vertex pair $(s,t)$ and failed edge set $F\subseteq E$ with $|F|\leq f$, the connectivity between $s$ and $t$ in $G-F$ can be answered by inspecting only the labels of $s$, $t$ and edges in $F$.

In this paper, we present a labeling scheme that uses $O(\log^{2}n)$-bit labels that can be computed in \emph{deterministic} polynomial time. This improves upon the previous $\tilde{O}(\sqrt{f})$ deterministic bound of Long, Pettie, and Saranurak \cite{long2025connectivity}, and even slightly improves the $O(\min\{f+\log n,\log^{2}n\log f\})$ \emph{randomized} bound of Dory and Parter \cite{DBLP:conf/podc/DoryP21} and Long, Pettie, and Saranurak \cite{long2025connectivity} when $f = \omega(\log^{2}n)$. Moreover, for a general $f$, this is the first labeling scheme that produces an $\tilde{O}(1)$-size labeling which is simultaneously correct across all queries.

Our approach combines the cycle-space-based labeling scheme from Dory and Parter \cite{DBLP:conf/podc/DoryP21} with a recent result by Knauer \cite{knauer2026logarithmic} on sparse cycle bases.
\end{abstract}

\newpage

\section{Introduction}

An \emph{edge-fault-tolerant  (EFT) connectivity labeling scheme} consists of a preprocessing algorithm and a query algorithm. Given an undirected graph $G = (V,E)$ and a fault bound $f$, the preprocessing algorithm computes a label function $L:V\cup E \to \{0,1\}^{*}$ that assigns each vertex and edge a binary string called a \emph{label}. Subsequently, for any pair of vertices $s,t\in V$ and failed edge set $F\subseteq E$ with $|F|\leq f$, the query algorithm can determine whether $s$ and $t$ are connected in $G-F$ by reading solely the labels of the query elements, namely, $L(s), L(t)$ and $\{L(e)\mid e\in F\}$. The \emph{label size} of a labeling scheme is the maximum label length of a label measured in bits. From a data structure perspective, this can be viewed as a distributed analog to \emph{edge-fault-tolerant connectivity oracles} \cite{patrascu2007planning}.

\paragraph{Randomized Labeling Schemes.} The fundamental work of Dory and Parter \cite{DBLP:conf/podc/DoryP21} presented the first two non-trivial EFT connectivity labeling schemes for general graphs, achieving label sizes of $O(f+\log n)$ based on cycle space sampling and $O(\log^3 n)$ based on graph sketching. Subsequently, Long, Pettie, and Saranurak \cite{long2025connectivity} slightly improved the second bound to $O(\log^{2}n\log f)$. All the aforementioned labeling schemes are Monte Carlo randomized, answering \emph{each} query correctly with high probability $1-1/n^{\Omega(1)}$. Namely, they can handle a polynomial number of queries with good success probability. However, these randomized constructions do not produce label functions that guarantee \emph{full} correctness across all possible queries simultaneously within the same label size bounds\footnote{Boosting the success probability of a randomized labeling scheme to $1 - n^{-\Omega(f)}$ (to cover all $n^{O(f)}$ possible queries) in a black-box way increases the label size by an $\tilde{O}(f)$ factor. The $O(f+\log n)$-size labeling scheme of Dory and Parter can guarantee full correctness after a white-box adaptation, albeit at the cost of increasing the label size to $O(f\log n)$.}. 

\paragraph{Deterministic Labeling Schemes.} Several subsequent works focus on EFT connectivity labeling schemes with full correctness, or even deterministic schemes with polynomial preprocessing time. In 2023, Izumi, Emek, Wadayama, and Masuzawa \cite{IzumiEWM23} showed a deterministic EFT connectivity labeling scheme with $\Otilde(f^{2})$ label size by developing a deterministic counterpart of graph sketching. Later, Long, Pettie, and Saranurak \cite{long2025connectivity} presented a deterministic labeling scheme with $\Otilde(\sqrt{f})$ label size using expander-based techniques. This currently represents the state-of-the-art bound for EFT connectivity labeling with full correctness, and it leaves a substantial gap compared to the $\widetilde{O}(1)$-size randomized schemes. In particular, the following question remains open.

\begin{center}
\textit{For a general fault bound $f$, is there an EFT connectivity labeling scheme}\\\textit{guaranteeing full correctness with $\Otilde(1)$ label size?}
\end{center}

\paragraph{Our Results.} We resolve this open problem affirmatively by showing such a labeling scheme with $O(\log^{2}n)$ label size. Moreover, our labeling scheme is deterministic with polynomial preprocessing time. See \Cref{thm:DetailedLabel} for a formal statement. Our result even slightly improves upon the state-of-the-art \emph{randomized} bounds of $\min\{f+\log n, \log^{2}n \log f\}$ \cite{DBLP:conf/podc/DoryP21,long2025connectivity} when $f = \omega(\log^2 n)$

Our labeling scheme builds upon the cycle-space-based scheme of Dory and Parter \cite{DBLP:conf/podc/DoryP21}. In particular, we exploit \emph{sparse cycle bases} to replace the cycle-space sampling in \cite{DBLP:conf/podc/DoryP21}, which is the only randomized part of their algorithm. Roughly speaking, a cycle basis is sparse if it has low edge congestion, i.e., each edge appears in only a small number of cycles in this basis. Very recently, Knauer \cite{knauer2026logarithmic} proves that a sparse cycle basis with $O(\log n)$ congestion must exist, improving upon the previous bound of $O(\log^{2}n)$ \cite{freedman2021building,lehner2026sparse}. As a by product, we present a polynomial-time deterministic algorithm that computes a sparse cycle basis with $O(\log n)$ congestion, complementing Knauer's existential result.

\section{Preliminaries}

Let $G = (V,E)$ be the input undirected graph. Without loss of generality, assume $G$ is simple and connected. Let $n$ and $m$ denote the number of vertices and edges in $G$.

\paragraph{Induced Edge Cuts.} For each subset of vertices $S\subseteq V$, let $\delta(S) = \{(u,v)\in E\mid u\in S,v\in V\setminus S\}$ be the set of crossing edges. We call such a $\delta(S)$ an \emph{induced edge cut} of $G$.

\paragraph{The Cycle Space.} For an edge set $C\subseteq E$, let $\deg_{C}(v)$ denote the number of $C$-edges incident to $v$ for each vertex $v\in V$.

The \emph{cycle space} of $G$ is
\[
\Zcal = \{C\subseteq E\mid \deg_{C}(v)\equiv 0\pmod 2 \text{ for each }v\in V\}.
\]
Every $C\in \Zcal$ is called a \emph{cycle}. A subset of cycles $\Bcal\subseteq \Zcal$ forms a \emph{cycle basis} if $\Bcal$ forms a basis of $\Zcal$ over $\F_{2}$. 

It is well known that when $G$ is connected, a cycle basis has size
\[
|\Bcal|= m-n+1.
\]
The following standard relation between induced edge cuts and cycle bases is also well-known, and we omit its proof.
\begin{fact}
\label{fact:CycleCutRelation}
Let $\Bcal$ be a cycle basis. An edge set $K\subseteq E$ is an induced edge cut of $G$ if and only if, for each cycle $C\in\Bcal$,
\[
|C\cap K|\equiv 0\pmod 2.
\]
Namely, $K$ and $C$ are orthogonal over $\F_2$.
\end{fact}


\section{The Labeling Scheme}

In this section, we will show an edge-fault-tolerant connectivity labeling scheme of label size $O(\log^{2}n)$ bits, proving \Cref{thm:DetailedLabel}.

\begin{theorem}
\label{thm:DetailedLabel}
There exists a pair of preprocessing algorithm and query algorithm satisfying the following.
\begin{itemize}
\item Given an undirected graph $G=(V,E)$, the preprocessing algorithm computes a label function $L:V\cup E\to \{0,1\}^{O(\log^{2}n)}$.
\item for any vertex pair $s,t\in V$ and failed edge set $F\subseteq E$, the algorithm can answer whether $s$ and $t$ are connected in $G-F$ by reading only the labels $L(s),L(t)$ and $\{L(e)\mid e\in F\}$. 
\end{itemize}
Both algorithms are deterministic. The preprocessing time is polynomial, and the query time is $O(|F|^{3}\log n)$.
\end{theorem}

Our approach largely follows that of Dory and Parter \cite{DBLP:conf/podc/DoryP21}, which is based on the cycle space. The key difference is that we exploit \emph{sparse} cycle bases. This replaces the cycle-space sampling in \cite{DBLP:conf/podc/DoryP21}, which is the only randomized part of their algorithm.

Let us first introduce sparse cycle bases. For a cycle basis $\Bcal$, its \emph{congestion} on each edge $e$ is 
\[
\gamma_{e}(\Bcal):= |\{C\in\Bcal\mid C\ni e\}|,
\]
i.e., the number of cycles in $\Bcal$ that contain $e$. The congestion of $\Bcal$ is naturally its maximum congestion over all edges, and a cycle basis is sparse, meaning that it has low congestion. 

Very recently, Knauer \cite{knauer2026logarithmic} showed the existence of sparse cycle bases with congestion $O(\log n)$. This already suffices to achieve a label size of $O(\log^{2}n)$ bits when preprocessing time is unrestricted. To further achieve deterministic polynomial preprocessing time, it requires an algorithmic sparse cycle basis (i.e., \Cref{thm:SparseCycleBasis}). We defer its proof to \Cref{sect:SparseCycleBasis}.

\begin{restatable}{theorem}{SparseCycleBasis}
\label{thm:SparseCycleBasis}
Given an undirected graph $G = (V,E)$, there exists a deterministic algorithm that computes a cycle basis of $G$ with congestion $O(\log n)$ in polynomial time.
\end{restatable}

In the rest of this section, we will prove \Cref{thm:DetailedLabel} using \Cref{thm:SparseCycleBasis}. First, in \Cref{sect:LinearSystem}, we interpret a query as testing the feasibility of a linear system related to a cycle basis. Next, in \Cref{sect:Label}, we leverage sparse cycle bases to construct compact labels that encode all necessary information for the linear system.

\subsection{The Linear System Interpretation}
\label{sect:LinearSystem}

This interpretation is already shown by \cite{DBLP:conf/podc/DoryP21}, and we outline it below for context. Fix a spanning tree $T$ and a cycle basis $\Bcal$ in $G$. Consider a query $\langle s,t,F\rangle$. Observe that, $s$ and $t$ are disconnected in $G - F$ if and only if there is an induced edge cut $K\subseteq F$ separating $s$ and $t$.
\begin{itemize}
\item By \Cref{fact:CycleCutRelation}, $K$ forms an induced edge cut exactly when the following linear constraints are satisfied.
\begin{equation}
\label{eq:cycle}
|C\cap K|\equiv 0\pmod 2,\qquad\forall C\in\Bcal
\end{equation}
\item Let $P$ be the unique $s$-$t$ path in the spanning tree $T$. Then $K$ separates $s$ and $t$ if and only if 
\begin{equation}
\label{eq:path}
|P\cap K|\equiv 1\pmod 2.
\end{equation}
That is, the $P$ crosses the cut $K$ an odd number of times.
\end{itemize}
In summary, it suffices to determine whether there exists a solution $K\subseteq F$ satisfying both (\ref{eq:cycle}) and (\ref{eq:path}).

\subsection{The Labels}

At a high level, when we pick a sparse cycle basis $\Bcal$ using \Cref{thm:SparseCycleBasis}, this linear system becomes sparse, and thus can be encoded into short labels. Formally, since the solution $K$ is restricted to being a subset of $F$, we can determine whether a feasible $K$ exists, once we can extract the sets 
\[
C_{F}:=C\cap F,\ \forall C\in\Bcal,\qquad\text{and}\qquad P_{F}:=P\cap F
\]
from the labels $L(s),L(t)$ and $\{L(e)\mid e\in F\}$. We now construct short labels to extract these sets $C_{F}$ and $P_{F}$.

\paragraph{Extract $C_{F}$.} In preprocessing, assign distinct IDs to the cycles in $\Bcal$. For each edge $e\in E$, store in $L(e)$ the sorted list of IDs of cycles containing $e$. At query time, group the occurrences of each cycle ID across the failed-edge labels: the failed edges listing that ID form precisely $C_F$. Cycles whose IDs never appear have $C_F=\emptyset$ and contribute only the equation $0=0$, so they can be omitted.

This part contributes $O(\log^{2}n)$ bits to each edge label: an edge belongs to $O(\log n)$ basis cycles, and each cycle ID takes $O(\log n)$ bits since $|\Bcal|\leq m=O(n^2)$.

\paragraph{Extract $P_{F}$.} This part follows \cite{DBLP:conf/podc/DoryP21}. Root $T$ arbitrarily and use standard $O(\log n)$-bit ancestry labels to determine whether a tree edge is ancestral to a vertex. Deleting a tree edge $e$ separates $T$ into its child subtree and the remaining vertices. Hence $e$ lies on the $s$-$t$ path $P$ exactly when it is ancestral to \emph{exactly one} of $s$ and $t$. Two ancestry tests therefore determine whether $e\in P_F$, using only the labels of $e,s,t$; non-tree edges never belong to $P_F$. This part contributes $O(\log n)$ bits to each label.

\paragraph{Analysis.} The total label size is $O(\log^{2}n)$ bits. Preprocessing is deterministic and polynomial by \Cref{thm:SparseCycleBasis}, together with a DFS traversal to construct the ancestry labels.

To check whether a suitable $K\subseteq F$ exists, introduce one variable $x_e\in\F_2$ for each $e\in F$, where $x_e=1$ means that $e\in K$. Then (\ref{eq:cycle}) and (\ref{eq:path}) become
\[
\sum_{e\in C_F}x_e=0\quad\text{for every }C\in\Bcal,
\qquad
\sum_{e\in P_F}x_e=1.
\]
These are linear equations over $\F_2$: the sets $C_F$ and $P_F$ extracted from the labels specify their coefficients. Every solution corresponds to a set $K=\{e\in F:x_e=1\}$ satisfying both conditions, and conversely. Thus Gaussian elimination tests whether such a $K$ exists; a feasible system means that $s$ and $t$ are disconnected.

The system has $|F|$ variables and $O(|F|\log n)$ constraints: besides the path equation, only cycles with nonempty $C_F$ contribute a constraint, and each failed edge belongs to $O(\log n)$ basis cycles. Standard Gaussian elimination therefore takes $O(|F|^3\log n)$ time, which dominates the query time.

\label{sect:Label}

\section{Algorithmic Sparse Cycle Basis}
\label{sect:SparseCycleBasis}
In this section, we will present a deterministic polynomial-time algorithm for computing sparse cycle bases with congestion $O(\log n)$. We follow the proof framework of Knauer's existential result \cite{knauer2026logarithmic}, and make each step of the argument algorithmic. The framework has three steps. \begin{enumerate}
\item Instead of directly computing cycle bases with low (maximum) congestion, we first aim at cycle bases with low \emph{weighted average} congestion. This algorithmic step is already shown by Rizzi \cite{DBLP:journals/algorithmica/Rizzi09}. We will give its formal statement in \Cref{sect:AverageSparseCycleBaiss}.
\item Next, providing a subroutine for computing cycle bases with low weighted average congestion, the standard multiplicative weight update (MWU) framework can compute a fractional solution of sparse cycle bases, i.e. a convex combination of bases that has low (maximum) congestion. This step is presented in \Cref{sect:MWU}.

\item Finally, by considering a linear matroid with cycles as ground elements, we can round the fractional solution into a single sparse cycle basis, using the swap-rounding technique of Chekuri, Vondr\'ak, and
Zenklusen \cite{chekuri2009dependent}. In fact, if randomization is allowed, it suffices to invoke the randomized swap rounding of \cite{chekuri2009dependent} as a black box. We can then derandomize it in our context using pessimistic estimators. This step is presented in \Cref{sect:DeterministicRound}.
\end{enumerate}

\subsection{Cycle Bases with Low Average Congestion}
\label{sect:AverageSparseCycleBaiss}

The following result is due to Rizzi \cite{DBLP:journals/algorithmica/Rizzi09}. See also Theorem 4.4 in the survey \cite{DBLP:journals/csr/KavithaLMMRUZ09} for a short proof. Informally, the idea is to iteratively find a simple cycle of length $O(\log n)$ and remove its heaviest edge. In fact, the algorithm further guarantees that the output cycle basis is weakly fundamental\footnote{A cycle basis $\{C_1, \dots, C_k\}$ is \emph{weakly fundamental} if there exists an ordering of the cycles such that each $C_i$ contains at least one edge not present in any preceding cycle $C_1, \dots, C_{i-1}$.}, but this property is not required for our application.

\begin{lemma}[Rizzi \cite{DBLP:journals/algorithmica/Rizzi09}, see also Theorem 4.4 in \cite{DBLP:journals/csr/KavithaLMMRUZ09}]
\label{lemma:WeightedBases}
Let $G = (V,E)$ be an undirected graph with edge weight $w:E\to \mathbb{R}_{\geq 0}$, there is a deterministic algorithm that computes a cycle basis $\Bcal$ s.t.
\[
\sum_{e\in E}w_{e}\gamma_{e}(\Bcal)\leq O(\log n)\cdot \sum_{e\in E}w_{e}.
\]
The algorithm runs in $O(mn)$ time.
\end{lemma}

\subsection{Fractional Sparse Cycle Bases via MWU}
\label{sect:MWU}

It is standard to use MWU to compute a good ``worst-case'' fractional solution given a subroutine for finding solutions that are ``good on average''. Readers familiar with this technique may safely skip this proof.

\begin{lemma}
\label{lemma:mwu}
Let $G$ be an undirected graph. There is a deterministic algorithm that computes $T=O(m\log n)$ many cycle bases $\Bcal_{1},...,\Bcal_{T}$, such that for each edge $e\in E$,
\[
\frac{1}{T}\cdot \sum_{t=1}^{T}\gamma_{e}(\Bcal_{t})\leq O(\log n).
\]
The algorithm runs in polynomial time.
\end{lemma}

\begin{algorithm}[ht]
\caption{Fractional Sparse Cycle Bases}\label{alg:mwu}
\begin{algorithmic}[1]
\State $w_e^{(1)}\gets 1$ for every $e\in E$.
\State Set $r=m-n+1$ be the size of a cycle basis. 
\State Set $T=r\cdot \lceil \ln m\rceil$
\For{$t=1,\ldots,T$}
  \State $\Bcal_t\gets$ the cycle basis from \Cref{lemma:WeightedBases} with weight $w^{(t)}$.
  \State $w_e^{(t+1)}\gets w_e^{(t)}\bigl(1+\gamma_{e}(\Bcal_{t})/r\bigr)$ for every $e\in E$.
\EndFor
\State \Return $\Bcal_{1},...,\Bcal_{T}$
\end{algorithmic}
\end{algorithm}

\begin{proof}
The algorithm is shown in \Cref{alg:mwu}. The algorithm clearly runs in deterministic polynomial time.

We now show its correctness. Let $W^{(t)}:=\sum_{e}w_{e}^{(t)}$ be the total weight at round $t$. Define $p^{(t)}_{e}:=w_{e}^{(t)}/W^{(t)}$, i.e. $p^{(t)}$ is the distribution over edges proportional to the edge weights $w^{(t)}$. Define 
\[
g_{e}^{(t)}:=\gamma_{e}(\Bcal_{t})/r
\]
to be the \emph{gain} of edge $e$ at round $t$. Note that $g_{e}^{(t)}\in [0,1]$ because the cycle basis $\Bcal_{t}$ has size at most $r$. The \emph{expected gain} at round $t$, defined as $\tilde{g}^{(t)}:=\sum_{e}p^{(t)}_{e}g^{(t)}_{e}$ is known to be small by \Cref{lemma:WeightedBases}. Formally, let $\alpha$ denote the $O(\log n)$ factor in \Cref{lemma:WeightedBases}. Then,
\[
\tilde{g}^{(t)}=\sum_{e}p^{(t)}_{e}g^{(t)}_{e} = \frac{\sum_{e}w^{(t)}_{e}\gamma_{e}(\Bcal_{t})}{W^{(t)}}\cdot \frac{1}{r}\leq \alpha/r.
\]

The total weight is closely related to the expected gain because of the multiplicative update rule. Formally, we have
\[
W^{(t+1)}=W^{(t)}+\tilde{g}^{(t)}\cdot W^{(t)} = W^{(t)}(1+\tilde{g}^{(t)}),
\]
and therefore
\[
W^{(T+1)}=W^{(1)} \cdot \prod_{t=1}^{T}(1+\tilde{g}^{(t)})\leq m\cdot\exp(\sum_{t}\tilde{g}^{(t)})\leq m\cdot\exp(T\alpha/r),
\]
where the second inequality uses $W^{(1)} = m$ and $\tilde{g}^{(t)}\in[0,1]$.

By the update rule again, we can now bound the congestion on each edge $e$.
\begin{align*}
w_{e}^{(T+1)} = \prod_{t=1}^{T}(1+\gamma_{e}(\Bcal_{t})/r) &\leq W^{(T+1)}\\
\sum_{t}\ln(1+\gamma_{e}(\Bcal_{t})/r)&\leq \ln m + T\alpha/r\\
\sum_{t}\gamma_{e}(\Bcal_{t})/r&\leq \frac{\ln m + T\alpha/r}{\ln 2},
\end{align*}
where the last inequality is due to $x\leq \frac{\ln(1+x)}{\ln 2}$ for any $x\in[0,1]$. Multiplying both side by $r/T$, we have
\[
\frac{1}{T}\sum_{t}\gamma_{e}(\Bcal_{t})\leq \frac{r\ln m}{T\ln 2} + \frac{\alpha}{\ln 2} \leq (\alpha+1)/\ln 2 = O(\log n)
\]
as desired

\end{proof}

\subsection{Deterministic Rounding: Proof of \Cref{thm:SparseCycleBasis}}
\label{sect:DeterministicRound}

The last step is to round the fractional solution of \Cref{lemma:mwu} to a sparse cycle basis. This completes the proof of \Cref{thm:SparseCycleBasis}. We restate \Cref{thm:SparseCycleBasis} below.

\SparseCycleBasis*

Define a matroid $\mathcal{M} = (\mathcal{Z}, \mathcal{I})$, where the ground set $\mathcal{Z}$ is the cycle space of $G$. Note that the matroid bases are exactly the cycle bases of $G$. 

\paragraph{Fractional Points of Fractional Cycle Bases.} Consider a convex combination of cycle bases
\[
\Bscr = \{(\lambda_{1},\Bcal_{1}),(\lambda_{2},\Bcal_{2}),...,(\lambda_{|\Bscr|},\Bcal_{|\Bscr|})\}
\]
where each $\lambda_{t}\geq 0$ and $\sum_{t=1}^{|\Bscr|}\lambda_{t}=1$. It corresponds to a fractional point
\[
\ppt(\Bscr):=\sum_{t=1}^{|\Bscr|}\lambda_{t}\vecone_{\Bcal_{t}}\in[0,1]^{\Zcal},
\]
where $\vecone_{\Bcal_{t}}$ is the characteristic vector of $\Bcal_{t}\subseteq \Zcal$. In particular, for a trivial convex combination $\Bscr = \{(1,\Bcal)\}$, $\ppt(\Bscr) = \vecone_{\Bcal}$.

For the convex combination obtained from \Cref{lemma:mwu}, denoted by
$\Bscr_{0}=\{(\frac{1}{T},\Bcal_{1}),...,(\frac{1}{T},\Bcal_{T})\}$ where $T = O(m\log n)$,
its corresponding fractional point is
$y_{0}=\ppt(\Bscr_{0}) = \sum_{t=1}^{T}\frac{1}{T}\vecone_{\Bcal_{t}}$.





\paragraph{The Potential Function.} At a high level, the rounding algorithm starts with the initial fractional basis $\Bscr_{0}$, iteratively updates it using \emph{swap operations}, and finally reaches a single (integral) cycle basis. The algorithm will ensure that the (fractional) basis always has low congestion along the way, which is measured by the following potential function. For a fractional point $y\in[0,1]^{\Zcal}$, define
\[
\Phi(y):= \sum_{e\in E}\prod_{\substack{C\in\Zcal\\\text{ s.t. }C\ni e}} (1+y_{C}).
\]
We will see that, the initial potential is small (\Cref{claim:InitialPotential}), the potential will not increase during the algorithm, and at the end, the single cycle basis with small potential must have low congestion (\Cref{claim:FinalPotential}). The monotonic non-increasing property of the potential will be shown later, after describing the algorithm.

\begin{claim}
\label{claim:InitialPotential}
The fractional point $y_{0}$ of the initial fractional basis $\Bscr_{0}$ has potential $\Phi(y_{0}) = \poly(n)$.
\end{claim}
\begin{proof}
Let $x=y_{0}$.
\[
\Phi(y_{0})= \sum_{e\in E}\prod_{\substack{C\in\Zcal\\\text{ s.t. }C\ni e}} (1+x_{C})\leq \sum_{e\in E}\prod_{\substack{C\in\Zcal\\\text{ s.t. }C\ni e}} \exp(x_{C}) = \sum_{e\in E}\exp\left(\sum_{t}\gamma_{e}\Bcal_{t}/T\right)\leq \poly(n).
\]
Here the first inequality uses $1+x\leq e^{x}$ for $x\geq 0$, and the last inequality uses \Cref{lemma:mwu}.
\end{proof}

\begin{claim}
\label{claim:FinalPotential}
Let $\Bcal$ be a cycle basis such that $\Phi(\vecone_{\Bcal})\leq \poly(n)$.  Then $\Bcal$ has congestion $\gamma(\Bcal)\leq O(\log n)$.
\end{claim}
\begin{proof}
By definition,
\[
\Phi(\vecone_{\Bcal})= \sum_{e\in E}2^{\gamma_{e}(\Bcal)}.
\]
Therefore, $\Phi(\vecone_{\Bcal})\leq \poly(n)$ implies that each edge $e$ has $\gamma_{e}(\Bcal)\leq O(\log n)$.
\end{proof}

\paragraph{The Algorithm.} The deterministic rounding algorithm is presented in \Cref{algo:DetRounding}

\begin{algorithm}[H]
\caption{Deterministic Rounding}
\label{algo:DetRounding}
\begin{algorithmic}[1]
\State $\Bscr \gets \Bscr_{0}$\Comment{The initial fractional basis from \Cref{lemma:mwu}}
\While{$|\Bscr|\geq 2$}
  \State Arbitrarily pick two (weighted) bases $(\alpha,\Acal)$ and $(\beta,\Bcal)$ in $\Bscr$.
  \While{$\Acal\neq \Bcal$}
    \State Arbitrarily pick a cycle $C\in \Acal\setminus \Bcal$.
    \State Find a cycle $D\in \Bcal\setminus \Acal$ s.t. both
           $\Acal-C+D$ and $\Bcal-D+C$ are cycle bases.
    \State $\Bscr_{1}\gets$ replace $(\alpha,\Acal)$ in $\Bscr$ with $(\alpha,\Acal-C+D)$
    \State $\Bscr_{2}\gets$ replace $(\beta, \Bcal)$ in $\Bscr$ with $(\beta,\Bcal - D + C)$
    \If {$\Phi(\ppt(\Bscr_{1}))\leq \Phi(\ppt(\Bscr_{2}))$}
        \State Update $\Bscr\gets \Bscr_{1}$, and now $\Acal$ refers to the new basis $\Acal-C+D$
    \Else
        \State Update $\Bscr\gets\Bscr_{2}$, and now $\Bcal$ refers to the new basis $\Bcal-D+C$.
    \EndIf
  \EndWhile
  \State Merge the two basis $(\alpha,\Acal),(\beta,\Bcal)\in \Bscr$ into a single $(\alpha+\beta,\Acal)$.
\EndWhile
\State \Return the unique basis in $\Bscr$.
\end{algorithmic}
\end{algorithm}

Let us first explain the algorithm in detail and analyse the running time along the way. Call each iteration of the outer loop a \emph{round}, and each iteration of the inner loop a \emph{swap operation}. In each round, the algorithm picks two different bases $\Acal$ and $\Bcal$ in $\Bscr$, and iteratively performs swap operations until they become the same. In fact, our swap operation is essentially identical to that in \cite{chekuri2009dependent}, except that we deterministically select the base ($\Bcal_{1}$ or $\Bcal_{2}$) with smaller potential rather than branching probabilistically.

A single swap operation takes polynomial time. In particular, the cycles $C$ and $D$ exist by the symmetric basis-exchange theorem for matroids, and we can find them in polynomial time by scanning all possible $D$ and testing whether $\Acal-C+D$ and $\Bcal-D+C$ are bases. In one round, the number of swap operations is at most $r=O(m)$ since each swap reduce the symmetric difference between $\Acal$ and $\Bcal$ by $1$. The number of rounds is $T = O(m\log n)$ since each round reduces the size of $\Bscr$ by $1$. Therefore, the whole algorithm takes polynomial time.

\paragraph{Correctness.} As discussed above, it suffices to show that the potential is non-increasing, which is established in the following \Cref{lemma:NonIncreasePotential}.

\begin{lemma}
\label{lemma:NonIncreasePotential}
In each swap operation,
\[
\min\{\Phi(\ppt(\Bscr_{1})),\Phi(\ppt(\Bscr_{2}))\}\leq \Phi(\ppt(\Bscr)).
\]
\end{lemma}
\begin{proof}
For clarity, let $y=\ppt(\Bscr),y_{1}=\ppt(\Bscr_{1})$ and $y_{2} = \ppt(\Bscr_{2})$ be the corresponding fractional points. For each edge $e$, let $\Zcal_{e}\subseteq \Zcal$ be the set of cycles that contain $e$, and define 
\[
\Phi_{e}(y):=\prod_{C\in \Zcal_{e}}(1+y_{C})
\]
be the term in the potential corresponding to $e$ ($\Phi_{e}(y_{1})$ and $\Phi_{e}(y_{2})$ are defined similarly). We will prove that
\begin{equation}
\label{eq:Inequality}
\frac{\beta}{\alpha+\beta}\Phi_{e}(y_{1}) + \frac{\alpha}{\alpha+\beta}\Phi_{e}(y_{2})\leq \Phi_{e}(y).
\end{equation}
Summing over all $e$ gives the desired inequality.

Recall that $C$ and $D$ denote the cycles picked by this swap operation. Compared to $y$, $y_1$ drops by $\alpha$ on entry $C$ and increases by $\alpha$ on entry $D$, while $y_{2}$ increases by $\beta$ on entry $C$ and drops by $\beta$ on entry $D$. Consider four cases. 
\medskip

\noindent\underline{Case 1.} Suppose neither $C$ nor $D$ belongs to $\Zcal_{e}$. Then clearly $\Phi_{e}(y_{1}) = \Phi_{e}(y_{2}) = \Phi_{e}(y)$.

\medskip

\noindent\underline{Case 2.} Suppose $C\in \Zcal_{e}$ and $D\notin \Zcal_{e}$. Note that $\Phi_{e}(y),\Phi_{e}(y_{1}),\Phi_{e}(y_{2})$ differ only in the terms corresponding to $C$. Since
\[
\frac{\beta}{\alpha+\beta}(1+y_{C}-\alpha) + \frac{\alpha}{\alpha+\beta}(1+y_{C}+\beta) = 1+y_{C},
\]
inequality (\ref{eq:Inequality}) follows.

\medskip

\noindent\underline{Case 3.} Suppose $D\in \Zcal_{e}$ and $C\notin \Zcal_{e}$. This case is symmetric to Case 2 and we omit its proof.

\medskip

\noindent\underline{Case 4.} Suppose $C,D\in \Zcal_{e}$. Then it suffices to prove
\[
\frac{\beta}{\alpha+\beta}(1+y_{C}-\alpha)(1+y_{D}+\alpha) + \frac{\alpha}{\alpha+\beta}(1+y_{C}+\beta)(1+y_{D}-\beta) \leq (1+y_{C})(1+y_{D}).
\]
This holds since the left-hand side is clearly $(1+y_{C})(1+y_{D})-\alpha\beta$.
\end{proof}

\section*{AI Disclosure}

The proof was initially discovered using ChatGPT 5.6 Sol. The authors have verified, simplified, and rewritten the proof. The authors assume all
responsibility for the paper’s content and correctness.

\bibliographystyle{alpha}
\bibliography{ref}

\end{document}